\documentclass[conference]{IEEEtran}
\IEEEoverridecommandlockouts
\newtheorem{theorem}{Theorem}[section]

\newtheorem{remark}{Remark}
\newtheorem{definition}{Definition}[section]

\newenvironment{proof}{\noindent {\bf Proof. }}{\hfill $\Box$ \\}

\usepackage{cite}
\usepackage{amsmath,amssymb,amsfonts}
\usepackage{algorithmic}
\usepackage{graphicx}
\usepackage{textcomp}
\usepackage{xcolor}

\begin{document}
\title{A Session Interaction Framework for The Multiple-Unicast Conjecture}


\author{Sirui Liu\textsuperscript{1,3}, Zongpeng Li\textsuperscript{3,1}, Xiying Fan\textsuperscript{2},
Haifeng Chen\textsuperscript{4}
\thanks{Corresponding author: Zongpeng Li, Email: zongpeng@tsinghua.edu.cn. This work has been supported by the Research Project of Quan Cheng Laboratory, China (Grant No. QCL20250108), the Shandong Provincial Natural Science Foundation, China (Grant No. ZR2025LZH008), the Research Project of Provincial Laboratory of Shandong, China (Grant No. SYS202201) and the Basic and Applied Basic Research Foundation of Guangdong Province (Grant No. 2023A1515110524). Xiying Fan is with University of Science and Technology Beijing (Beijing \& Shunde Innovation School, Guangdong).} \\
\thanks{%
\textcopyright~2026 IEEE.
Personal use of this material is permitted.
Permission from IEEE must be obtained for all other uses,
in any current or future media, including reprinting/republishing
this material for advertising or promotional purposes,
creating new collective works, for resale or redistribution
to servers or lists, or reuse of any copyrighted component
of this work in other works.%
}
\textsuperscript{1}\emph{Institute for Network Sciences and Cyberspace, Tsinghua University, Beijing, China} \\
\textsuperscript{2}\emph{University of Science and Technology Beijing, Beijing, China} \\
\textsuperscript{3}\emph{Quan Cheng Laboratory, Shandong, China}\\
\textsuperscript{4}\emph{China Unicom, Beijing, China}
}

\maketitle


\begin{abstract}
The multiple-unicast conjecture asserts that network coding offers no throughput advantage over routing in undirected networks. Its validity is known to imply fundamental lower bounds in computational complexity. We propose a \textit{Session Interaction Framework} that reduces the conjecture to a central equivalence: the conjecture holds universally if and only if every irreducible core is independent. This result transforms the global feasibility problem into a two-stage process. First, to make the reduction phase tractable, we provide simplified sufficient conditions for session dominance, offering geometric criteria to iteratively simplify complex session sets. Second, for the remaining ``irreducible core'', we propose a \textit{Session Decoupling Theorem}, reducing the conjecture's validity for a session set to its independent subsets. Topologically, we prove that sessions separated by high-cost cuts or cut-vertices are guaranteed to be independent. By integrating these reduction and decomposition mechanisms, our framework offers a systematic methodology to verify the conjecture across general network topologies.
\end{abstract}


\section{Introduction}

Network coding allows intermediate nodes to mix information streams~\cite{ahlswede2000network}. While its capacity benefits are well-established in single-session scenarios---such as multicast, the optimal solution for multiple unicast sessions remains undetermined. In this context, the routing-based solution corresponds to the standard multicommodity flow problem, yet the structure of the optimal network coding solution remains elusive.

The multiple-unicast conjecture (Li and Li~\cite{li2004network}, Harvey {\em et al.}~\cite{harvey2004comparing}) asserts that for undirected networks, network coding offers no throughput advantage over routing. Despite its seemingly simple nature, this conjecture has withstood over two decades of scrutiny, widely regarded as ``arguably the most important open problem in the field''~\cite{adler2006capacity}. Recently, its significance has transcended information theory: Farhadi {\em et al.}~\cite{10.1145/3313276.3316337} and Afshani {\em et al.}~\cite{afshani2019lower} revealed that its validity implies breakthrough lower bounds in computational complexity, making its resolution pivotal for delineating the limits of both information flow and computation.

Prior studies~\cite{6283627,6750094,seymour1980four,hu1963multi,s3-42.1.178,198499,liu2025isit,9729815,8101552,jain2006capacity,adler2006capacity,liu2026multiple} have predominantly focused on global network properties (\emph{e.g.} the cut-set bound) or information inequalities, often encountering combinatorial intractability as network complexity grows. To address this, our recent work~\cite{liu2026multiple} introduced \textit{Session Dominance}, a structural reduction concept. While dominance theoretically preserves the conjecture's validity through simplification, two critical challenges remain. First, finding an optimal reduction sequence based on the general definition is combinatorially non-trivial. Second, dominance is inherently a reduction mechanism---it leaves an ``irreducible core'' of sessions whose mutual entanglement remains uncharacterized. Thus, a complete framework requires not only conducting reduction but also resolving the interaction nature of the remaining core.

In this work, we propose the \emph{Session Interaction Framework} to achieve this. We introduce \textit{Session Independence} as the fundamental theoretical counterpart to dominance. By integrating these concepts, we address the global feasibility problem via a verifiable two-stage process:

\textbf{Stage I: Operational Reduction via Dominance.} We transcend the general definition in~\cite{liu2026multiple} by deriving tractable geometric criteria for dominance. We show that sessions satisfying conditions such as \emph{Cycle Closure} and \emph{Non-Crossing} can be deterministically reduced, bypassing the need for complex combinatorial search.

\textbf{Stage II: Decomposition via Independence.} For the remaining ``irreducible core,'' we examine whether sessions are information-theoretically entangled. We prove that the conjecture holds \emph{if and only if} every irreducible core is independent. Topologically, we show that features like high-cost cuts serve as certificates for this independence.

Our main contributions are summarized as follows:
\begin{enumerate}
    \item \textbf{Constructive Reduction Criteria:} We provide new, tractable sufficient conditions (Cycle Closure and Non-Crossing) to identify session dominance efficiently.
    \item \textbf{Interaction Characterization Theorem:} We introduce Session Independence as the theoretical counterpart to dominance, proving that the conjecture holds globally if and only if every irreducible core is independent.
    \item \textbf{Topological Decoupling:} We establish the Session Decoupling Theorem and identify topological independence certificates---specifically high-cost cuts and cut-vertices---to enable rigorous problem decomposition.
\end{enumerate}
\section{Notations and Preliminaries}

\subsection{Network Model}
We model the network as an undirected graph $G = (V, E)$ with edge capacities $\mathbf{c} \in \mathbb{R}_{\ge 0}^{|E|}$ and rate vector $\mathbf{r}$. We consider $k$ unicast sessions $\mathcal{H} = \{ (s_i, t_i) \}_{i=1}^k$, where $s_i, t_i \in V$ are the source and sink, and $r_i$ is the demand rate. To handle directed flows, each undirected edge $e=\{u,v\}$ is replaced by opposing arcs $\vec{uv}$ and $\vec{vu}$. A feasible coding flow $f$ must satisfy the aggregate capacity constraint $f(\vec{uv}) + f(\vec{vu}) \le c(e)$ for every edge $e \in E$.

\subsection{The Multiple-Unicast Conjecture}
The multiple-unicast conjecture asserts that for undirected networks, network coding offers no throughput advantage over routing. We analyze this via its dual formulations:

\textbf{Throughput Domain.} 
A demand rate vector $\mathbf{r}$ is achievable with network coding if and only if it is achievable with routing.

\textbf{Cost Domain.} 
By linear programming duality, the conjecture is equivalent to verifying that network coding cannot reduce the total transmission cost. We denote by $J_{\text{code}}(\mathcal{H})$ and $J_{\text{route}}(\mathcal{H})$ the minimum total cost $\sum_{e \in E} w(e)f(e)$ required to satisfy the demands in $\mathcal{H}$ using network coding and routing, respectively. For any edge weight vector $\mathbf{w} \in \mathbb{R}_{\ge 0}^{|E|}$, any valid coding flow $f$ satisfying demand $\mathbf{r}$ must adhere to:
\begin{equation} \label{eq:conjecture_cost}
    \sum_{e=(u,v) \in E} w(e) \left( f(\vec{u v}) + f(\vec{v u}) \right) \ge \sum_{i=1}^k r_i \cdot d_w(s_i, t_i),
\end{equation}
where $d_w(u, v)$ is the shortest-path distance under $\mathbf{w}$. We focus on verifying Eq.~\eqref{eq:conjecture_cost}.

\subsection{Session Dominance~\cite{liu2026multiple}}
To analyze Eq.~\eqref{eq:conjecture_cost}, we utilize the \textit{Session Dominance} framework, which reduces complex instances based on metric properties.

\begin{definition}[Dominance] \label{def:dominance}
    A session $h_1=(s_1, t_1)$ \emph{dominates} $h_2=(s_2, t_2)$ in $(G, \mathbf{w})$ if $d_w(s_1, s_2) + d_w(s_2, t_2) + d_w(t_2, t_1) = d_w(s_1, t_1)$.
    
    An instance $\mathcal{I}'=(G', \mathcal{H}', \mathbf{w}')$ dominates $\mathcal{I}=(G, \mathcal{H}, \mathbf{w})$ if:
    \begin{enumerate}
        \item $G \subseteq G'$ and $\mathbf{w}'|_G \le \mathbf{w}$;
        \item For every session $(s, t) \in \mathcal{H}$, the shortest-path distance $d_{w'}(s, t)$ in $G'$ equals $d_w(s, t)$ in $G$;
        \item The set $\mathcal{H}'$ dominates $\mathcal{H}$ in $G'$.
    \end{enumerate}
\end{definition}

\begin{theorem} \label{thm:reduction}
If an instance $\mathcal{I}'$ dominates $\mathcal{I}$, and the cost equality $J_{code} = J_{route}$ holds for $\mathcal{I}'$, then it also holds for $\mathcal{I}$.
\end{theorem}
 
We identify dominance using tractable pairwise criteria.

\begin{theorem}[Pairwise Dominance] \label{thm:pairwise}
    Two sessions $h_1, h_2$ can be dominated by a single session (or a symmetric pair) if they satisfy either:
    \begin{enumerate}
        \item $d_w(s_1, t_1) + d_w(s_2, t_2) \le d_w(s_1, t_2) + d_w(s_2, t_1)$;
        \item $d_w(s_1, t_1) + d_w(s_2, t_2) \le d_w(s_1, s_2) + d_w(t_1, t_2)$.
    \end{enumerate}
\end{theorem}

\section{The Session Interaction Framework}

The Session Interaction Framework adopts a hierarchical reduction strategy. We propose that the complexity of the multiple-unicast conjecture can be resolved through a sequential verification process. First, Reduction via Dominance utilizes distance-based properties defined in Section II to identify and merge redundant sessions, thereby strictly reducing the size of the session set. Second, Decomposition via Independence verifies if the remaining sessions are non-interacting; if so, the global problem is decomposed into independent sub-problems.

In this section, we first provide the formal definition of Session Independence in the Cost Domain. With the concepts of reducibility and independence established, we present the \textbf{Interaction Characterization Theorem}, which serves as the theoretical foundation for this framework.

\subsection{Session Independence}

Complementing dominance, \emph{Session Independence} characterizes the absence of coupling between sessions. We define this property in the \textbf{Cost Domain} to facilitate our reduction framework.

Let $\mathcal{H}$ be the set of sessions and $\mathcal{M} = \{M_1, \dots, M_k\}$ be the set of independent source messages. We model the transmission on any edge $e$ as a sequence of random variables (symbols).

\begin{definition}[Session Independence] 
\label{def:independence}
    Given a network instance $(G, \mathcal{H}, \mathbf{w})$, the session set $\mathcal{H}$ is \emph{independent}\footnote{In the conference version, the existential quantifier over the minimum-cost coding solution in this definition was inadvertently scoped separately for each session pair. The intended statement, used by the subsequent proofs and in the conference presentation, is the common-solution formulation stated here.} under $\mathbf{w}$ if there exists a minimum-cost coding solution $\mathcal{C}^{\star}$ such that, for every pair of distinct sessions $h_i,h_j \in \mathcal{H}$, every edge $e \in E$, and every symbol $X$ transmitted on $e$ by $\mathcal{C}^{\star}$,
    \begin{equation}
    \label{eq:independence}
        I(X; M_i \mid \mathbf{M}_{\setminus \{i\}}) \cdot I(X; M_j \mid \mathbf{M}_{\setminus \{j\}}) = 0,
    \end{equation}
    where $\mathbf{M}_{\setminus \{k\}} = \mathcal{M} \setminus \{M_k\}$.
    This implies that a single symbol $X$ cannot simultaneously carry information for both session $i$ and session $j$.

    More generally, two disjoint subsets $\mathcal{H}_1,\mathcal{H}_2 \subseteq \mathcal{H}$ are \emph{independent} under $\mathbf{w}$ if there exists a minimum-cost coding solution for $\mathcal{H}$ such that Eq.~\eqref{eq:independence} holds for every $h_i \in \mathcal{H}_1$, every $h_j \in \mathcal{H}_2$, every edge $e \in E$, and every symbol $X$ transmitted on $e$ by that solution.
\end{definition}

\begin{theorem}[Equivalence to Session Independence]
\label{thm:equivalence_independence}
For a given network instance $(G, \mathcal{H}, \mathbf{w})$, the cost equality $J_{code}(\mathcal{H}) = J_{route}(\mathcal{H})$ holds if and only if $\mathcal{H}$ is independent under $\mathbf{w}$.
\end{theorem}

\begin{proof}
    We prove the equivalence in two directions.

    \noindent\textbf{1. Necessity ($\Rightarrow$):}
    Assume the cost equality holds ($J_{\text{code}} = J_{\text{route}}$).
    Then, an optimal routing solution serves as a minimum-cost coding solution in which information streams remain disjoint.
    In this solution, any transmitted symbol $X$ depends on at most one message $M_k$.
    Consequently, for $i \neq k$, $I(X; M_i | M_{\setminus \{i\}}) = 0$.
    This ensures the product in Eq.~\eqref{eq:independence} vanishes for all pairs, satisfying Definition \ref{def:independence}.

    \noindent\textbf{2. Sufficiency ($\Leftarrow$):} 
    Assume $\mathcal{H}$ is independent. By Definition \ref{def:independence}, let $f$ be the flow vector of a minimum-cost coding solution for $\mathcal{H}$ that witnesses independence.
    The independence condition implies that the information flows for different sessions are information-theoretically orthogonal on every edge. Consequently, the required flow rate $f_e$ is strictly additive, decomposing into a sum of single-session flows: $f_e = \sum_{k} f_e^{(k)}$, where $f_e^{(k)}$ depends only on message $M_k$.
    
    Since each component $f^{(k)}$ is a feasible solution for the single unicast session $h_k$, its cost contribution must be at least the optimal coding cost for that session ($J_{\text{code}}(h_k)$). Thus, the total cost satisfies:
    \begin{equation}
        J_{\text{code}}(\mathcal{H}) = \sum_{e} w(e)\sum_{k} f_e^{(k)} \ge \sum_{k} J_{\text{code}}(h_k).
    \end{equation}
    Since network coding provides no cost advantage for single unicast ($J_{\text{code}}(h_k) = J_{\text{route}}(h_k)$) and routing costs are additive:
    \begin{equation}
        J_{\text{code}}(\mathcal{H}) \ge \sum_{k} J_{\text{route}}(h_k) = J_{\text{route}}(\mathcal{H}).
    \end{equation}
    Given that $J_{\text{code}}(\mathcal{H}) \le J_{\text{route}}(\mathcal{H})$ always holds, we conclude: $J_{\text{code}}(\mathcal{H}) = J_{\text{route}}(\mathcal{H}).$
\end{proof}

\subsection{Interaction Characterization}

We now combine dominance and independence to analyze the conjecture. First, we define {\em reducibility}.

\begin{definition}[Reducible Pair]
\label{def:reducible}
    A pair of sessions $h_i, h_j$ is \emph{reducible} if it satisfies the Pairwise Dominance condition (Theorem \ref{thm:pairwise}) and is not symmetric ({\em i.e.}, the sessions are not merely reversed).
    An \emph{Irreducible Core} admits no further reducible pairs. By Theorem \ref{thm:reduction}, verifying the cost equality on any resulting core suffices for the original instance, regardless of the reduction sequence.
\end{definition}

\begin{theorem}[Interaction Characterization]
\label{thm:interaction}
The multiple-unicast conjecture holds for all undirected networks if and only if, for every network instance and every non-negative edge-weight vector $\mathbf{w}$, every irreducible core is independent under $\mathbf{w}$.
\end{theorem}

\begin{proof}
    Let $m$ be the number of terminal nodes. We use induction on $m$, with the base cases $m \le 4$ known to hold~\cite{seymour1980four}.

    \noindent\textbf{1. Sufficiency ($\Leftarrow$):}
    Consider an instance with $m > 4$.
    \begin{itemize}
        \item \textbf{Case 1: A Reducible Pair exists.} We apply dominance reduction. Since the pair is reducible (non-symmetric), the reduction strictly decreases the number of terminals to $m' < m$. By the inductive hypothesis, the conjecture holds for the reduced instance, and thus for the original (Theorem \ref{thm:reduction}).
        \item \textbf{Case 2: No Reducible Pair exists.} 
        Fix an arbitrary edge-weight vector $\mathbf{w}$. The remaining session set is an irreducible core and is therefore independent under $\mathbf{w}$. By Theorem \ref{thm:equivalence_independence}, this independence guarantees $J_{code} = J_{route}$ under $\mathbf{w}$. Since $\mathbf{w}$ was arbitrary, the conjecture follows.
    \end{itemize}

    \noindent\textbf{2. Necessity ($\Rightarrow$):}
    Assume the conjecture holds universally. Suppose there exists an irreducible core that is not independent under some edge-weight vector $\mathbf{w}$. By Theorem \ref{thm:equivalence_independence}, we must have $J_{code} < J_{route}$ under this $\mathbf{w}$, which directly contradicts the universal validity of the conjecture. Thus, every irreducible core must be independent under every non-negative edge-weight vector.
\end{proof}
\subsection{Stage I: Reduction via Dominance}
\noindent \textbf{The Challenge of Reduction Order.} The concept of \emph{mergeable sessions}\cite{liu2026multiple} suggests that a session set can be recursively merged into a symmetric pair. However, the merging order is critical. In general graphs, a suboptimal merging sequence may lead to a ``dead end," necessitating an intractable search to find a valid reduction path.

To bypass this search, we identify a sufficient condition on cycles that guarantees order-independence. We show that satisfying \emph{Cycle Closure} and \emph{Non-Crossing} ensures that the set can be deterministically reduced.

\begin{theorem}[New constructive criteria]
\label{thm:cost_dominance_condition}
Given an instance $(G,\mathcal{H},\mathbf{w})$. Let $C$ be a cycle in $G$. Consider a subset of sessions $\mathcal{H}_C \subseteq \mathcal{H}$ whose terminals all lie on $C$. If the following two conditions are met:
\begin{enumerate}
    \item \textbf{Cycle Closure:} The shortest path between any two terminals in $\mathcal{H}_C$ lies entirely on $C$;
    \item \textbf{Non-Crossing:} The sessions in $\mathcal{H}_C$ are pairwise non-crossing on the cycle.
\end{enumerate}
Then, the entire set $\mathcal{H}_C$ is dominated by a symmetric pair of sessions.
\end{theorem}

\begin{proof}We iteratively merge session pairs satisfying Theorem \ref{thm:pairwise}. Consider the terminals of $\mathcal{H}_C$ arranged clockwise on $C$. By the pairwise non-crossing property, scanning this sequence yields a consecutive subsequence of terminals belonging to two distinct sessions $h_i$ and $h_j$. Let the terminals of $h_i$ be $\{u_i, v_i\}$ and those of $h_j$ be $\{u_j, v_j\}$. These terminals must appear in one of two basic topological orders: \textbf{sequential} ($u_i, v_i, u_j, v_j$) or \textbf{nested} ($u_i, u_j, v_j, v_i$).

We verify that in both cases, the shortest paths connecting the corresponding terminals satisfy the dominance inequality. Due to the \textit{Cycle Closure} condition, the shortest path between any two terminals lies entirely on the cycle $C$. For the \textbf{sequential case}, we consider the cross-connection paths $u_i \leadsto u_j$ and $v_i \leadsto v_j$. For the \textbf{nested case}, we consider the paths $u_i \leadsto v_j$ and $u_j \leadsto v_i$.
Due to the cyclic ordering, these chosen paths inevitably traverse the intermediate terminals ({\em e.g.}, in the sequential case, the path $u_i \leadsto u_j$ on the cycle must pass through $v_i$ or $v_j$). Consequently, the sum of the lengths of these ``crossed" paths covers the segments of the original paths $u_i \leadsto v_i$ and $u_j \leadsto v_j$ plus the gap between them. This geometric property implies that the sum of the original distances is less than or equal to the sum of the crossed distances, thereby satisfying Condition (1) or (2) of Theorem \ref{thm:pairwise}, depending on the specific assignment of source ($s$) and sink ($t$) roles to $u$ and $v$.

In both scenarios, the pair $\{h_i, h_j\}$ satisfies one of the dominance inequalities. We now explicitly construct a dominating instance. Let $h_i = (s_i, t_i)$ and $h_j = (s_j, t_j)$.

Assume $d_w(s_i, t_i) + d_w(s_j, t_j) \le d_w(s_i, t_j) + d_w(s_j, t_i)$. We extend $G$ with auxiliary nodes $S, T$ and edges $(S, s_k), (t_k, T)$ with weights $x_k, y_k$ for $k \in \{i, j\}$. To ensure $(S, T)$ dominates $\{h_i, h_j\}$, we require the length of the path passing through matching terminals ($S \to s_k \to t_k \to T$) to be a constant $L$, while the path through crossed terminals must have length at least $L$. This yields the system:
$
x_k + y_k + d_w(s_k, t_k) = L, x_i + y_j + d_w(s_i, t_j) \ge L \quad k \in \{i, j\}
.$
By treating $x_k, y_k, L$ as variables, it can be shown that non-negative solutions exist if and only if the distance inequality of Condition (1) holds. Furthermore, by selecting a sufficiently large $L$, we ensure that shortest paths between other terminals in $G$ remain unaffected. Thus, $\{h_i, h_j\}$ is replaced by $(S, T)$.

Assume the inequality takes the form $d_w(s_i, t_i) + d_w(s_j, t_j) \le d_w(s_i, s_j) + d_w(t_j, t_i)$. The construction is analogous, but we connect $S$ to $s_i$ and $t_j$, and $T$ to $s_j$ and $t_i$. The resulting dominating set contains the symmetric pair $\{(S, T), (T, S)\}$. The solvability condition for the corresponding linear system is exactly Condition (2).

We employ a sequential merging strategy. After replacing the initial pair $\{h_i, h_j\}$ with the dominating session $h_{new}=(S, T)$ (and replacing the corresponding terminal subsequence $\{u, v, u, v\}$ in the ordering), we turn to the session $h_k$ whose terminals are consecutive to or immediately enclose the segment occupied by $\{h_i, h_j\}$. Fig.~\ref{fig:nested_reduction} illustrates this critical step: the dominating session $h_{new}$ (red) is constructed to replace the adjacent pair $\{h_i, h_j\}$.

The new pair $\{h_{new}, h_k\}$ is valid for the next merger step for two reasons. First, since $h_{new}$ inherits the position previously held by $h_i$ and $h_j$, it maintains the non-crossing relationship with $h_k$. Second, regarding the dominance inequality, the auxiliary weights for $h_{new}$ are explicitly constructed to preserve the shortest-path distances. Since the original session $h_k$ satisfied the pairwise dominance condition with both $h_i$ and $h_j$ (as established in our initial analysis), the condition is preserved for $\{h_{new}, h_k\}$.

Consequently, we can merge $h_{new}$ and $h_k$ into a new session. We repeat this process, continuously identifying adjacent pairs in the contiguous subsequence, and sequentially absorbing sessions one by one. This iteration continues until the entire set $\mathcal{H}_C$ is reduced to a single symmetric pair that dominates the original session set.
\end{proof}
\begin{figure}[t]
    \centering
    \includegraphics[width=0.48\linewidth]{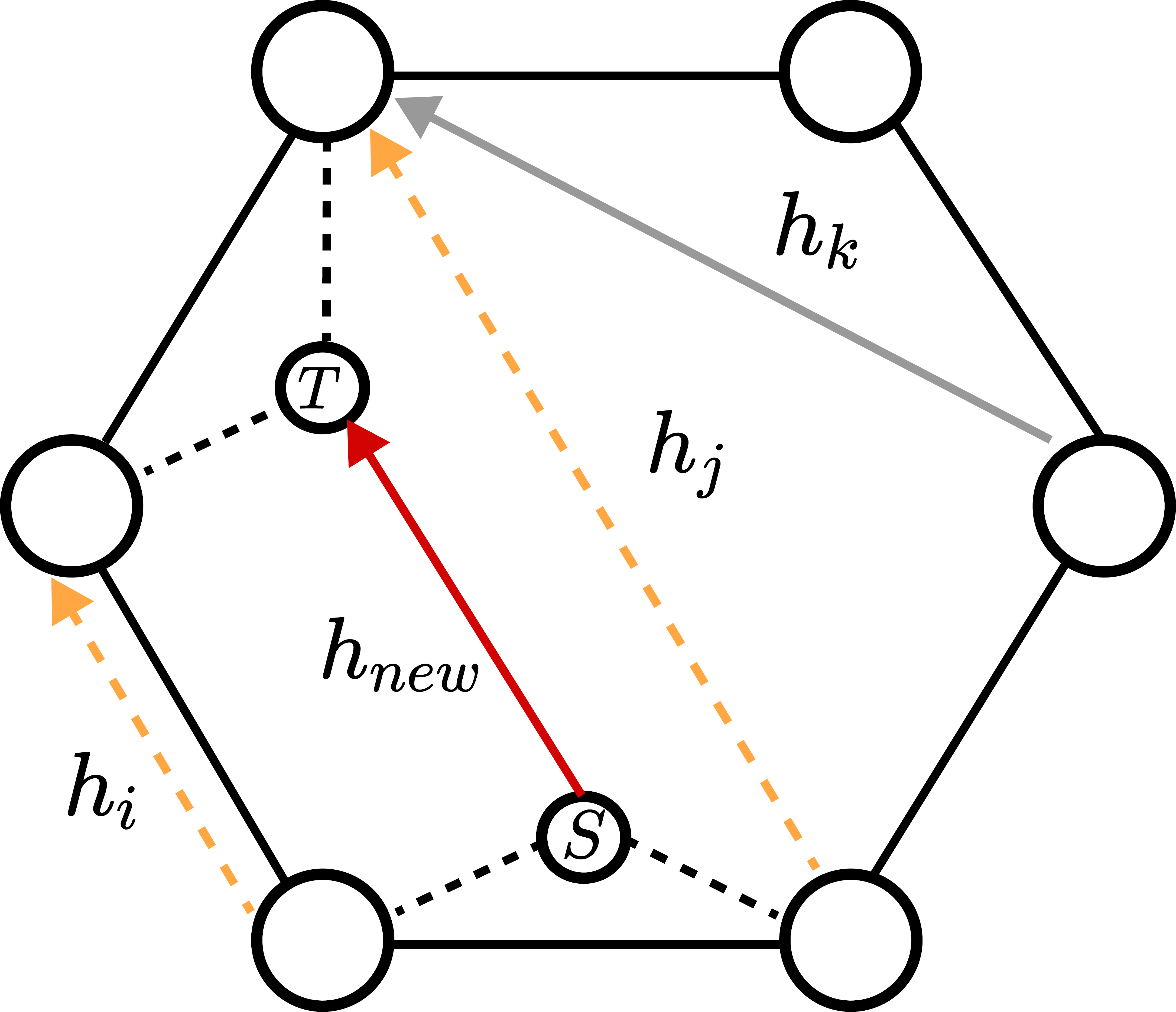} 
    \caption{Geometric construction for the reduction (Theorem \ref{thm:cost_dominance_condition}). The adjacent pair $\{h_i, h_j\}$ (dashed orange) is dominated by the session $h_{new}=(S, T)$ (solid red) via auxiliary nodes. }
    \label{fig:nested_reduction}
\end{figure}
\begin{theorem}
\label{thm:throughput_conjecture_condition}
Let $G=(V, E)$ be an undirected graph with edge capacities $c$, and let $\mathcal{H}$ be a set of sessions. Suppose $\mathcal{H}$ can be partitioned into two subsets $\mathcal{H}_1$ and $\mathcal{H}_2$. For each subset $\mathcal{H}_i$ ($i \in \{1, 2\}$), let $V_{\mathcal{H}_i}$ denote its terminal nodes. Consider a transformed graph $G'_i$ obtained by recursively applying the following edge contraction operations specific to $\mathcal{H}_i$:
\begin{enumerate}
    \item Contract an edge $(u, v)$ if $u, v \notin V_{\mathcal{H}_i}$ (merging two non-terminals);
    \item Contract an edge $(u, v)$ if $u \in V_{\mathcal{H}_i}, v \notin V_{\mathcal{H}_i}$ (merging a terminal with a non-terminal).
\end{enumerate}
If, for each $i$, the resulting graph $G'_i$ forms a simple cycle $C_i$ such that the terminals of $\mathcal{H}_i$ map to distinct nodes on $C_i$ in a pairwise non-crossing order, then the multiple-unicast conjecture holds for the set of sessions $\mathcal{H}$.
\end{theorem}

\begin{proof}
We establish the validity of the multiple-unicast conjecture by invoking the duality principle established by Li and Li ~\cite{li2004network}. The conjecture holds for the instance $(G, \mathcal{H})$ with arbitrary edge capacities $c$ if and only if it holds for the instance $(G, \mathcal{H})$ under any non-negative edge weight vector $w$. Thus, it suffices to prove that the cost inequality (Eq. (\ref{eq:conjecture_cost})) holds under an arbitrarily chosen, fixed weight vector $w \in \mathbb{R}_{\ge0}^{|E|}$. Let us proceed by evaluating the network under this specific $w$.

For each subset $\mathcal{H}_i$, we construct a corresponding cycle structure in the original graph $G$. Let $C_i$ be the cycle in the contracted graph $G'_i$. We restore a cycle $C'_i$ in $G$ by replacing each edge $(u, v)$ on $C_i$ with the shortest path $P_{uv}$ between these nodes in $G$ under weights $w$.

We first verify the \textbf{Cycle Closure} condition. We proceed by contradiction. Suppose that for some pair of terminals, the shortest path in $G$ deviates from the restored cycle $C'_i$. This deviation would imply the existence of a path structure in $G$ that connects terminals through nodes or edges not accounted for by the cycle geometry. Upon applying the edge contraction operations defined in the theorem, such a structure would not collapse into a simple cycle $C_i$; instead, it would result in a graph $G'_i$ containing additional branches or multiple paths between nodes ({\em e.g.}, a theta graph). This contradicts the hypothesis that $G'_i$ is a simple cycle. Therefore, the shortest path between any two terminals in $\mathcal{H}_i$ must lie entirely on $C'_i$.

We next verify the \textbf{Non-Crossing} condition. The restoration process expands edges into paths but preserves the relative ordering of the nodes. Since the terminals of $\mathcal{H}_i$ are arranged in a pairwise non-crossing order on the contracted cycle $C_i$ by hypothesis, this non-crossing sequence is maintained on the restored cycle $C'_i$.

By Theorem \ref{thm:cost_dominance_condition}, each subset $\mathcal{H}_i$ is dominated by a symmetric pair under the chosen $w$. Thus, under $w$, the original instance is dominated by an instance with at most four terminals. Since the cost inequality (Eq. (\ref{eq:conjecture_cost})) holds for all 4-terminal instances under any valid weight vector\cite{seymour1980four}, it holds for $\mathcal{H}$ under $w$. By the arbitrariness of $w$, Eq. (1) holds globally for all non-negative weight vectors. Through the duality\cite{li2004network}, the multiple-unicast conjecture is verified for $\mathcal{H}$.
\end{proof}

\subsection{Stage II: Decomposition via Independence}

While dominance reduces the complexity of specific sessions, \emph{independence} allows us to break down the global problem structure. We propose the \textbf{Session Decoupling Theorem} to leverage this property for problem decomposition.

\begin{theorem}[Session Decoupling]
\label{thm:decoupling}
 Let $\mathcal{H}$ be a session set. If $\mathcal{H}$ can be partitioned into disjoint subsets $\mathcal{H}_1$ and $\mathcal{H}_2$ that are independent under a specific weight vector $\mathbf{w}$, then the cost equality $J_{code} = J_{route}$ holds for $\mathcal{H}$ under $\mathbf{w}$ if and only if it holds for both $\mathcal{H}_1$ and $\mathcal{H}_2$ under $\mathbf{w}$.
\end{theorem}

\begin{proof}
    By Definition \ref{def:independence}, let $f$ be the flow of a minimum-cost coding solution witnessing the independence of $\mathcal{H}_1$ and $\mathcal{H}_2$. Eq.~\eqref{eq:independence} implies $f$ decomposes into orthogonal flows $f = f_1 + f_2$, where $f_k$ supports $\mathcal{H}_k$. Substituting this into the cost objective: $J_{\text{code}}(\mathcal{H}) = \sum_{e} w(e)f_1(e) + \sum_{e} w(e)f_2(e)$. Since $f_k$ is a feasible solution for $\mathcal{H}_k$, the term $\sum_{e} w(e)f_k(e)$ is lower-bounded by the optimal cost, {\em i.e.}, $\sum_{e} w(e)f_k(e) \ge J_{\text{code}}(\mathcal{H}_k)$. Thus, $J_{\text{code}}(\mathcal{H}) \ge J_{\text{code}}(\mathcal{H}_1) + J_{\text{code}}(\mathcal{H}_2)$. Given the additivity of routing costs ($J_{\text{route}}(\mathcal{H}) = \sum J_{\text{route}}(\mathcal{H}_k)$), the validity of the conjecture for subsets implies $J_{\text{code}}(\mathcal{H}) \ge J_{\text{route}}(\mathcal{H})$. Since $J_{\text{code}}(\mathcal{H}) \le
J_{\text{route}}(\mathcal{H})$ always holds, equality follows.
The converse follows immediately by contraposition: if the cost
equality fails for either subset, combining a minimum-cost coding
solution for that subset with an optimal routing solution for the
other subset yields
$J_{\text{code}}(\mathcal{H})<J_{\text{route}}(\mathcal{H})$.
\end{proof}

To apply this theorem practically, we identify topological conditions that guarantee independence. The core mechanism is \emph{Metric Decomposition}~\cite{8101552}: we decompose the network $(G, \mathbf{w})$ into two auxiliary networks $(G, \mathbf{w}_1)$ and $(G, \mathbf{w}_2)$ such that $\mathbf{w} = \mathbf{w}_1 + \mathbf{w}_2$. If we ensure that shortest path distances satisfy $d_{\mathbf{w}}(s_i, t_i) = d_{\mathbf{w}_1}(s_i, t_i) + d_{\mathbf{w}_2}(s_i, t_i)$ for all sessions, the linearity of coding cost implies that a network coding advantage in $G$ exists only if it exists in $G_1$ or $G_2$.
\begin{theorem}[Separation by High-Cost Cut]
\label{thm:high_cost_cut}
Let a cut $(V_U, V_V)$ partition the session set $\mathcal{H}$ into $\mathcal{H}_U$ (in $V_U$) and $\mathcal{H}_V$ (in $V_V$). Let $E_{cut}$ be the set of edges crossing the cut.
If there exist two reference nodes $u \in V_U$ and $v \in V_V$ (with no direct edge $(u, v)$) such that for every cut edge $e=(x,y) \in E_{cut}$ (with $x \in V_U, y \in V_V$), the edge weight satisfies $\frac{1}{2} w(x, y) \ge \max \{ d_w(u, x), d_w(v, y) \}$, then the session sets $\mathcal{H}_U$ and $\mathcal{H}_V$ are independent.
\end{theorem}

\begin{proof}
We prove independence by explicitly constructing the metric decomposition $\mathbf{w} = \mathbf{w}_1 + \mathbf{w}_2$.
To construct $\mathbf{w}_1$, we aim to ``collapse" the $V_V$ side to the single node $u$. For internal edges, we preserve weights in $V_U$ and set weights in $V_V$ to zero. For the cut edges, if an edge connects to $u$, we set $w_1(e) = 0$; if it connects to $v$, we maintain the full weight $w_1(e) = w(e)$; for all other cut edges, we set $w_1(e) = \frac{1}{2}w(e)$. (The condition that no edge $(u,v)$ exists ensures this assignment is well-defined).
Symmetrically, to construct $\mathbf{w}_2$, we collapse $V_U$ to $v$. We set internal weights in $V_U$ to zero and preserve weights in $V_V$. For cut edges, if an edge connects to $v$, set $w_2(e) = 0$; if it connects to $u$, set $w_2(e) = w(e)$; otherwise, set $w_2(e) = \frac{1}{2}w(e)$. Summing these constructions yields $\mathbf{w}_1 + \mathbf{w}_2 = \mathbf{w}$ for all edges.

In the network $(G, \mathbf{w}_1)$, since edges in $V_V$ and incident to $u$ are zero-cost, all nodes in $V_V$ effectively shrink to $u$. For $\mathcal{H}_U$, we must ensure shortest paths remain within $V_U$ without taking ``shortcuts" through the cut. Any path attempting to utilize the cut must exit $V_U$ at some node $x$ to enter the zero-cost region $V_V$. If $x \neq u$, the crossing cost is at least $\frac{1}{2}w(x,y)$. The high-cost condition $\frac{1}{2}w(x,y) \ge d_w(u,x)$ implies that routing directly to $u$ (the gateway to the zero-cost region) is always cheaper or equal to crossing at $x$. This confirms that no optimal flow for $\mathcal{H}_U$ would traverse these high-cost edges, thus preserving distances ($d_{\mathbf{w}_1} = d_{\mathbf{w}}$).
Simultaneously, for every session in $\mathcal{H}_V$, the distance becomes zero. These trivial sessions can be eliminated. This is because removing them does not alter the routing cost, and if a hypothetical coding scheme utilized these zero-cost sessions to achieve a gain for $\mathcal{H}_U$, one could derive a superior scheme for the pruned instance via hard-coding.
Thus, the analysis of the conjecture on $(G, \mathbf{w}_1)$ effectively reduces to $\mathcal{H}_U$ alone. By symmetry, $(G, \mathbf{w}_2)$ reduces to $\mathcal{H}_V$.
This decomposition implies that in the constructed auxiliary instances, the High-Cost Cut carries zero flow. Consequently, the sessions $\mathcal{H}_U$ and $\mathcal{H}_V$ are decoupled.
\end{proof}
\begin{remark}The high-cost cut in Theorem \ref{thm:high_cost_cut} guarantees independence only under a specific $w$. Since Theorem 3.2 requires independence for all $w \in \mathbb{R}_{\ge0}^{|E|}$, this serves merely as a local metric certificate. Establishing independence for arbitrary weights requires topological features, such as those introduced in the following theorem.
\end{remark}
\begin{theorem}[Separation by Cut-Vertex]
\label{thm:cut_vertex}
Let $v$ be a cut-vertex in $G$ whose removal partitions the vertices into disjoint sets $V_1$ (containing terminals of $\mathcal{H}_1$) and $V_2$ (containing terminals of $\mathcal{H}_2$). Then, the session sets $\mathcal{H}_1$ and $\mathcal{H}_2$ are independent.
\end{theorem}

\begin{proof}
This follows as a special case of Theorem \ref{thm:high_cost_cut}. We construct the metric decomposition $\mathbf{w} = \mathbf{w}_1 + \mathbf{w}_2$.
For $\mathbf{w}_1$, we set weights in $V_2$ to zero, collapsing the entire component to node $v$. Consequently, distances for $\mathcal{H}_2$ vanish, allowing their elimination.
For $\mathcal{H}_1$, since $v$ is the unique vertex connecting $V_1$ and $V_2$, any path leaving $V_1$ must traverse $v$. As $v$ acts as the boundary to the zero-cost region $V_2$, extending a path into $V_2$ and returning to $V_1$ merely adds non-negative weight without offering a shortcut. Thus, shortest paths for $\mathcal{H}_1$ are preserved.
Symmetry holds for $\mathbf{w}_2$. Since this independence condition relies solely on topology and holds for any weight vector $\mathbf{w}$, by the duality between cost and throughput domains, the sessions are independent for any edge capacity vector $\mathbf{c}$.
\end{proof}
\section{Conclusion}

We present the Session Interaction Framework to dissect the conjecture via a two-stage hierarchical reduction. Locally, Cycle Closure and Non-Crossing conditions enable deterministic reduction, bypassing complex searches. Globally, Metric Decomposition identifies cuts that decouple independent parts. Together, these tools transform complex graphs into verifiable structures. Future work involves integrating this toolbox with information inequalities to validate broader graph classes.




\bibliographystyle{IEEEtran}
\bibliography{ref}

\end{document}